\documentclass[12pt]{article} 
\usepackage{a4wide}

\usepackage{amssymb,amsmath,tikz}
\usetikzlibrary{shapes.misc, fit, decorations.pathreplacing,calligraphy,positioning} 
\usepackage{amsthm}
\usetikzlibrary{calc}
\usepgflibrary{arrows}
\usepackage{caption,subcaption}
\usepackage{verbatim,mathdots,stmaryrd}
\usepackage{rotating}
\usepackage{float}
\usepackage{array}
\usepackage{changepage} 
\usepackage{wasysym} 
\usepackage{graphicx}
\usepackage{xurl}
\usepackage{thm-restate}

\newtheorem{proposition}{Proposition}
\newtheorem{theorem}{Theorem}

\theoremstyle{definition}
\newtheorem{definition}{Definition}
\theoremstyle{remark}

\newcommand{\M}{\mathcal{M}}
\renewcommand{\phi}{\varphi}
\newcommand{\dist}{D}
\newcommand{\lang}{\mathcal{L}}
\newcommand{\atoms}{\mathcal{P}}
\newcommand{\agents}{\mathcal{A}}
\newcommand{\R}[3]{\mathcal{R}_{#1,#2,#3}}
\newcommand{\true}[1]{\llbracket #1 \rrbracket}
\newcommand{\ats}[1]{\ulcorner #1 \urcorner}

\newcommand{\nonlang}{\ensuremath\cap}
\newcommand{\yeslang}{\ensuremath\lang_0}
\newcommand{\single}{\ensuremath\odot}
\newcommand{\multi}{\mathchoice
  {\mathbin{\vcenter{\hbox{\scalebox{1.6}{$\bullet$}}}}} 
  {\mathbin{\vcenter{\hbox{\scalebox{1.6}{$\bullet$}}}}} 
  {\mathbin{\vcenter{\hbox{\scalebox{1.2}{$\bullet$}}}}} 
  {\mathbin{\vcenter{\hbox{\scalebox{1.2}{$\bullet$}}}}} 
}
\newcommand{\simult}{\ensuremath{\Uparrow}}
\newcommand{\seq}{\ensuremath\omega}
\newcommand{\ord}{\ensuremath\Omega}
\newcommand{\some}{\ensuremath\exists}
\newcommand{\all}{\ensuremath\forall}
\newcommand{\na}{\ensuremath\epsilon}

\newcommand{\nlsingle}{{\ensuremath{(\nonlang, \na, \na, \some)}}}
\newcommand{\nlall}{{\ensuremath{(\nonlang, \na, \na, \all)}}}

\newcommand{\formulaseqsingle}{{\ensuremath{(\yeslang, \single, \omega, \some)}}}
\newcommand{\formulaseqall}{{\ensuremath{(\yeslang, \single, \omega, \all)}}}
\newcommand{\formulasimsingle}{{\ensuremath{(\yeslang, \single, \simult, \some)}}}
\newcommand{\formulasimall}{{\ensuremath{(\yeslang, \single, \simult, \all)}}}

\newcommand{\setseqsingleomega}{{\ensuremath{(\yeslang, \multi, \omega, \some)}}}
\newcommand{\setseqallomega}{{\ensuremath{(\yeslang, \multi, \omega, \all)}}}
\newcommand{\setseqsingleOmega}{{\ensuremath{(\yeslang, \multi, \Omega, \some)}}}
\newcommand{\setseqallOmega}{{\ensuremath{(\yeslang, \multi, \Omega, \all)}}}
\newcommand{\setsimsingle}{{\ensuremath{(\yeslang, \multi, \simult, \some)}}}
\newcommand{\setsimall}{{\ensuremath{(\yeslang, \multi, \simult, \all)}}}

\usepackage{soul}
\usepackage{todonotes}

 \newcommand{\rev}[1]{{#1}}

\title{Varieties of Distributed Knowledge\thanks{This work was presented at \textit{Advances in Modal Logic} (AiML) 2024.}}
\date{}
\author{Rustam Galimullin\thanks{University of Bergen, Bergen, Norway; \texttt{rustam.galimullin@uib.no}}
\and
Louwe B. Kuijer\thanks{University of Liverpool, Liverpool, U.K.; \texttt{Louwe.Kuijer@liverpool.ac.uk}}
}

\begin{document}

\maketitle
  \begin{abstract}
  Distributed knowledge is one of the better known group knowledge modalities. While its intuitive idea is relatively clear, there is ample room for interpretation of details. 
  We investigate 12 definitions of distributed knowledge that differ from each other in the kinds of information sharing the agents can perform in order to achieve shared 
  \rev{mutual} knowledge of a proposition.
  We then show which kinds of distributed knowledge are equivalent, and  which kinds imply each other, i.e., for any two variants $\tau_1$ and $\tau_2$ of distributed knowledge we show whether a proposition $\phi$ being distributed knowledge under definition $\tau_1$ implies that $\phi$ is distributed knowledge under definition $\tau_2$.

    \end{abstract}


\section{Introduction}
\label{sec:intro}
Epistemic logic (see, e.g., \cite{fagin95,meyer_1995}) can be used to describe the knowledge of one or more agents. If multiple agents are involved, one can then study various kinds of group knowledge. On the one hand, we may consider types of group knowledge that are stronger than individual knowledge; for example, one may wonder whether a particular proposition $\phi$ is known by all members of the group, or even whether $\phi$ is so obvious (to the group members) as to be common knowledge. On the other hand, we can also consider a type of group knowledge that is weaker than individual knowledge; even if $\phi$ is not currently known by any individual group member, the group might be able to learn $\phi$ if they combine their information. For example, perhaps agent $a$ knows that $\phi\rightarrow \psi$ and $b$ knows $\phi$. Neither of them knows $\psi$, yet if they pool their knowledge they would be able to get to know it.

This latter kind of group knowledge is typically known as \emph{distributed knowledge} (see, e.g., \cite{halpern90,meyer_1995,fagin95,vanderhoek99,roelofsen07,AgotnesWang2017}). 
Distributed knowledge is a kind of hypothetical knowledge: $\phi$ is distributed knowledge among a group $G$ if the members of $G$ could, if they combined their knowledge, learn that $\phi$ is true. Note that we do not require the agents to actually combine their knowledge in this way, it suffices that they \emph{could} do so and learn $\phi$.

While the general idea of distributed knowledge is reasonably clear, there is no consensus about how to formally define it. Broadly speaking, there have historically been two main approaches.
We will refer to these approaches as the \emph{intersection} approach and the \emph{full communication} approach, with the latter term being derived from \cite{vanderhoek99}. We should note, however, that these are not standardised terms. In fact, even ``distributed knowledge'' is not fully standard, with other terms such as ``group knowledge'', ``collective knowledge'' and ``implicit knowledge'' also being used.

In both approaches, the distributed knowledge of a group $G$ of agents depends on what information the group members possess. In epistemic logic, the information state of an agent $a$ is generally represented by an accessibility relation $\sim_a$, and $a$ knows a formula $\phi$, denoted $\square_a\phi$, in \rev{world} $s$ if and only if $\phi$ is true in every \rev{world} $s'$ such that $s\sim_a s'$.

The intersection approach is the most common one, and is used in \cite{halpern90,meyer_1995,fagin95,vanderhoek99,gerbrandy99,roelofsen07,AgotnesWang2017}, among many others. Here, $\phi$ is distributed knowledge in \rev{world} $s$ if $\phi$ is true in every \rev{world} $s'$ such that $s\sim_G s'$, where $\sim_G = \bigcap_{a\in G}\sim_a$. The intuition behind this approach is that 
the group $G$ is collectively capable of distinguishing between $s$ and $s'$ if any of its members can.

The full communication approach is less common, but still used in many places, including \cite{humberstone85,vanderhoek99,gerbrandy99,roelofsen07}. In this approach, $\phi$ is distributed knowledge among $G$ if and only if the set of formulas known by any of the agents entails $\phi$, i.e., if $\{\psi\in \lang_0\mid \exists a\in G: s\models\square_a\psi\}\models \phi$. In order to avoid circularity we do have to be careful to specify that the known formulas $\psi$ must not 
reference distributed knowledge, i.e. they must be from the basic epistemic logic $\lang_0$.

Observant readers may notice that there is significant overlap between the list of papers using the intersection approach and those using the full communication approach. This is because one of the topics studied has been the relation between the intersection and full communication versions of distributed knowledge. The outcomes of this comparison are that (i) if $\phi$ is full communication distributed knowledge then it is also intersection distributed knowledge, (ii) $\phi$ can be intersection intersection distributed knowledge without being full communication distributed knowledge, and (iii) on certain types of models, the two kinds of distributed knowledge are equivalent.

The reason the two approaches to distributed knowledge persist side by side, albeit with the intersection approach being more popular, is that they appeal to different intuitions. Specifically, the issue is whether agents share information in a way that can be expressed in epistemic logic. Suppose that $a$ considers a \rev{world} $s_1$ possible but $s_2$ impossible, while $b$ considers $s_1$ to be impossible and $s_2$ to be possible, but that $s_1$ and $s_2$ are not distinguishable by any formula of epistemic logic. Can $a$ and $b$, when working together, discover that neither $s_1$ nor $s_2$ is possible?

The full communication approach says ``no, they cannot exclude $s_1$ and $s_2$''. After all, while $a$ does not consider $s_2$ to be possible, there is nothing they can say to $b$ that would communicate this impossibility, and $b$ is likewise incapable of communicating the impossibility of $s_1$.
The intersection approach, on the other hand, says ``yes, they can exclude $s_1$ and $s_2$''. Even if neither of them can express the difference (in epistemic logic, at least), $a$ knows that $s_2$ is impossible and $b$ knows that $s_1$ is impossible, so together they know both are impossible. Perhaps they communicate this impossibility to each other in a language other than epistemic logic, such as first-order logic. Perhaps they simply point at the \rev{worlds} they consider impossible. Perhaps they perform a Vulcan mind-meld, or somehow merge their databases or neural networks. What matters, to the intersection approach, is not how the agents share their information, but only that the agents possess the required information.

In this paper, we will not try to settle the debate in favour of one variant. On the contrary, we will introduce several further variants of distributed knowledge. This is because, in addition to the form of information sharing (formulas or mind-meld) which makes the difference between the intersection and full communication variants, there are several more questions one can ask about how distributed knowledge is established.



How much information do the agents share? Do they share information simultaneously, or is there an order? Are all agents required to know that $\phi$ is true after the knowledge sharing, or does it suffice if one agent knows $\phi$?

Different answers to these questions may lead to different notions of distributed knowledge. That is not to say that all possible combinations lead to different kinds of distributed knowledge. For example, suppose that each agent shares a single proposition known to that agent. Then it does not matter whether the agents share simultaneously or in order (Proposition~\ref{prop:finite_sequence}), and if one of the agents can learn $\phi$ then all of them can learn it (Proposition~\ref{prop:finite_singleall}). But in other cases, the difference does matter.

We will consider 12 possible definitions of distributed knowledge. One of these, which we label $\nlall$, is the intersection definition of distributed knowledge. None of our definitions is exactly the same as the full communication definition, but the variant that we label $\formulasimall$ is equivalent to full communication (Proposition \ref{prop:connections}).\footnote{Furthermore, several other variants are equivalent to $\formulasimall$, and therefore, by transitivity, also equivalent to full communication.} \rev{To the best of our knowledge, other definitions of distributed knowledge in our taxonomy have not been considered in the literature before.} 

After introducing the basic technical definitions in Section \ref{sec:definitions},  we will define all variants of distributed knowledge in  Section~\ref{sec:variants} and compare them to the existing approaches. Then we compare the introduced definitions of distributed knowledge to each other in Section \ref{sec:expressivity}. Finally, we conclude and outline the directions for further research in Section \ref{sec:conclusion}. 


\section{Basic Definitions}
\label{sec:definitions}
Each of the variants of distributed knowledge that we consider will use the same language, which is basic epistemic logic with an additional operator $D_G$, indicating distributed knowledge among group $G$.
\begin{definition}
Let $\atoms$ be a countable set of propositional atoms and $\agents$ a finite set of agents. The language $\lang$ is given by the following normal form
\begin{equation*}
\lang \ni \phi ::= p \mid \neg \phi \mid (\phi \vee \phi) \mid \square_a\phi\mid D_G \phi
\end{equation*}
where $p\in \atoms$, $a\in \agents$ and \rev{$\emptyset \neq G \subseteq \agents$}. We denote the fragment of $\lang$ that does not contain $D_G$ by $\lang_0$. We omit parentheses where this should not cause confusion.

We use $\wedge, \rightarrow, \leftrightarrow$ and $\lozenge_a$ in the usual way as abbreviations. Similarly, we use $\bigwedge$ and $\bigvee$ for $n$-ary conjunction and disjunction, respectively. 
\end{definition}

Because we are describing (distributed) \emph{knowledge}, we will use S5 models. We should stress, however, that our results also hold for K models.
\begin{definition}
A \emph{model} $\M$ is a tuple $(S,\sim,V)$, where $S$ is a non-empty set of \rev{\emph{worlds}}, $\sim:\agents \rightarrow 2^{S\times S}$ assigns to each agent $a\in \agents$ an equivalence relation $\sim_a \subseteq S\times S$ and $V:\atoms \rightarrow 2^S$ is the valuation function. If necessary, we will refer to the elements of the tuple as $S^\M$, $\sim^\M$, and $V^\M$.
A \emph{pointed model} is a pair $\M,s$ where $s$ is a \rev{world} of $\M$. 
\end{definition}

All operators other than $D_G$ are given their normal semantics. The semantics of $\lang_0$ is therefore as follows.
\begin{definition}
Let $\M=(S,\sim,V)$ be a model and $s\in S$. Then

\begin{tabular}{lll}
$\M,s\models p$& iff & $s\in V(p)$\\
$\M,s\models \neg \phi$ & iff &$\M,s\not\models \phi$\\
$\M,s\models \phi \vee \psi$ & iff & $\M,s\models \phi$ or $\M,s\models \psi$\\
$\M,s\models \square_a\phi$ & iff & $\M,t\models \phi$ for all $t$ such that $s\sim_a t$.
\end{tabular}
\end{definition}
For $D_G$, the semantics will depend on the type of distributed knowledge under consideration, which we discuss in the next section.

In several of the proofs throughout this paper, we will make use of the concept of $Q$-bisimilarity, where $Q\subseteq \atoms$. 

\begin{definition}
Let $Q \subseteq \atoms$, and $\M = (S^\M, \sim^\M, V^\M)$ and $\mathcal{N} = (S^\mathcal{N}, \sim^\mathcal{N},$ $V^\mathcal{N})$ be models. We say that $\M$ and $\mathcal{N}$ are $Q$-\emph{bisimilar} (denoted $\M \approx_Q \mathcal{N}$) if there is a non-empty relation $B \subseteq S^\M \times S^\mathcal{N}$, called $Q$-\emph{bisimulation}, such that for all $B(s,t)$, the following conditions are satisfied:
\begin{description}
\item[\textbf{Atoms}] for all $p \in Q$: $s \in V^\M(p)$ if and only if $t \in V^\mathcal{N}(p)$,
\item[\textbf{Forth}] for all $a \in \agents$ and  $u \in S^\M$ such that $s \sim^\M_a u$, there is a $v \in S^\mathcal{N}$ such that $t \sim^\mathcal{N}_a v$ and $B(u,v)$,
\item[\textbf{Back}] for all $a \in \agents$ and $v \in S^\mathcal{N}$ such that $t \sim^\mathcal{N}_a v$, there is a $u \in S^\M$ such that $s \sim^\M_a u$ and $B(u,v)$. 
\end{description} 
We say that $\M,s$ and $\mathcal{N},t$ are $Q$-bisimilar and denote this by $\M,s \approx_Q \mathcal{N},t$ if there is a $Q$-bisimulation linking \rev{worlds} $s$ and $t$.
\end{definition}

In the paper, we will make use of the classic result that bisimilar models satisfy the same formulas of epistemic logic. 

\begin{theorem}
\label{thm:equiv}
Given $\M,s$ and $\mathcal{N},t$, if $\M,s \approx_Q \mathcal{N},t$, then for all  $\varphi \in \lang_0$ that include atoms only from $Q$, we have that $\M,s \models \varphi$ if and only if $\mathcal{N},t \models \varphi$.
\end{theorem}

\section{Varieties of distributed knowledge}
\label{sec:variants}
In Section~\ref{sec:intro} we mentioned a number of questions regarding the exact workings of distributed knowledge. Here we discuss these questions in more detail, and use the potential answers to define types of distributed knowledge. Before we get into these details, however, we should remark on one aspect of distributed knowledge that will hold for every variant, namely that distributed knowledge is backward-looking.

In both the intersection and full communication definitions of distributed knowledge,
a proposition $\phi$ is distributed knowledge among $G$ if, by combining their knowledge, $G$ can discover that $\phi$ \emph{was} true \emph{before} they combined their knowledge.\footnote{See \cite{AgotnesWang2017} for a more thorough discussion of this aspect of distributed knowledge.} In this paper we also follow this tradition. The past tense is important because $\phi$ may contain claims that certain group members are ignorant of some fact, and this ignorance can be broken when agents in $G$ share their knowledge.  

For example, let $\phi$ be the famous Moore-sentence $p \wedge \neg\square_a p$, i.e., ``$p$ is true but $a$ does not know that $p$ is true.'' 
This sentence cannot be known by agent $a$, yet it can be distributed knowledge between $a$ and $b$. Perhaps $a$ knows that $\neg \square_ap$ while $b$ knows that $p$. When $a$ and $b$ combine their knowledge, they will learn that $p\wedge \neg \square_ap$ used to be true, but that very same communication will render the formula false, since $a$ will learn that $p$ is true. 

Because of this backward looking nature, distributed knowledge is a \emph{static} operator, as opposed to the \emph{dynamic} operators from dynamic epistemic logic (DEL) \cite{del}. A dynamic take on distributed knowledge is also possible, and would likely correspond to what agents may learn through communication with each other\footnote{As with static distributed knowledge, we need to account for many small but important implementation decisions, for example the extent to which agents that are not part of the group will be aware of the discussion among the group members.}. Some of the known approaches include a single agent sharing all her information with everyone \cite{Baltag2010slides}, a group of agents sharing everything they know among themselves \cite{AgotnesWang2017,baltag20}, topic-based communication within a group of agent\rev{s} \cite{galimullin23}, and various forms of public communication by agents and their effects \cite{agotnes10,agotnes22,galimullin23b}.
While such dynamic treatment of distributed knowledge 
is interesting, it is outside the scope of this contribution.

We now continue with a detailed discussion of each of our questions regarding the meaning of distributed knowledge.

\paragraph{Forms of information}
The first important question is the form of information shared by the group members in their attempt to establish knowledge of a proposition.
We consider two answers to this question. Firstly, agents may be able to share \emph{formulas of $\lang_0$ that they know}. So if $\square_a\psi$ holds, for some $\psi\in \lang_0$ and $a\in G$, then in a group discussion among $G$, agent $a$ can contribute $\psi$. The restriction to formulas that agents actually know arises from the intuition that distributed knowledge is about combining individual \textit{knowledge} rather than arbitrary formulas.  

The alternative is that the agents may be sharing information in a way that is either entirely non-lingual, or at least phrased in a language stronger than $\lang_0$. Importantly, such information sharing is not bound to respect bisimilarity in the sense of Theorem \ref{thm:equiv} (see, e.g., \cite{roelofsen07}).

Note that we restrict $\psi$ to $\lang_0$, so the formulas that the agents can share in their deliberation cannot include the $D_G$ operator. 
This is required in order to avoid vicious circularity; if we allow the $D_G$ operator to be used during the deliberations, there are situations where agents can learn $\phi$ if they combine their knowledge, but only if that knowledge includes the fact that $\phi$ is distributed knowledge. Hence $\phi$ would be distributed knowledge if and only if\dots{} $\phi$ is distributed knowledge. As a result of this circularity, the semantics would become underdetermined, i.e., there would be pointed models where both $D_G\phi$ and $\neg D_G\phi$ are consistent with the semantics (see Appendix \ref{app:vicious}).

Recall from the introduction that the intersection definition of distributed knowledge assumes the non-lingual answer to this question, whereas the full communication definition assumes the information being communicated is in the form of $\lang_0$ formulas.


\paragraph{Amount of information}
The next question is how much information the agents share. In principle, any measure could be used here. For example, one could imagine a situation where each agent has, say, 5 seconds to contribute their share. Or perhaps agents are limited to statements of a given maximum complexity.

Here, however, we will restrict ourselves to a coarser distinction: agents will be able to share either a single formula, or an infinite set of formulas. Note that since we are limiting only the amount of formulas, not their complexity, it would not make sense to restrict to a given finite number of formulas, 
since any finite number of formulas can be combined into one using conjunctions. 

Note that this distinction only makes sense if information is shared as $\lang_0$ formulas; if information is shared non-linguistically we do not have a sensible measure of the amount of information shared.


\paragraph{Order and turn-taking}
Another consideration is whether the agents share all information simultaneously, or in some order (with agents taking a single turn if they share one formula, or multiple turns if they share a set of formulas). This distinction is relevant because agents can only share formulas that they know; if the agents share their information in some order, then the later agents may be able to contribute some formulas that they did not know initially but that they have come to know based on the information provided by the agents before them.
Again, this distinction only works if agents share their information in the form of formulas.

\paragraph{Collective or individual success}
Finally, we can distinguish between a type of distributed knowledge where all group members need to learn the truth of a formula and a type where only one group member needs to learn it. In other words, if there is a possible communication between $a$ and $b$ that would result in $a$ knowing $\phi$ while $b$ remains ignorant of it, would $\phi$ be distributed knowledge?

\paragraph{Equivalence among variants}
Not all of the combinations of answers to questions from the previous section make sense. 
Still, the answers allow us to define 12 variations of distributed knowledge. 

We should stress, however, that not all these variations are truly different. For example, suppose that, in their communication, every agent shares a single formula, and they do so simultaneously. Then it is possible for the agents to communicate in such a way that a single agent learns $\phi$ if and only if it is possible for them to communicate in a way where all agents learn $\phi$ (see Proposition~\ref{prop:single}).

In fact, our main contributions in this paper are (1) formal definitions of the various kinds of distributed knowledge, and (2) the results on which of the variants are equivalent to each other.



\subsection{Semantics for distributed knowledge}

In our taxonomy, a type of distributed knowledge can be identified by the form of information shared, the amount of information, whether there is an order, and how many agents need to learn the target formula. A type $\tau$ is therefore a tuple $\tau = (f,a,o,q)$, where $f\in \{\nonlang, \yeslang\}$ indicates whether the information is presented in the form of formulas ($\lang_0$) or not ($\cap$), $a\in \{\single, \multi, \na\}$ indicates whether a single formula is shared ($\single$) or a set of formulas ($\multi$), while $a = \na$ is used for the case where $f=\nonlang$ and therefore no formulas are shared at all. The parameter $o\in \{\simult, \seq, \ord, \na\}$ indicates whether the agents share their knowledge simultaneously ($\simult$), in a sequence with a length $\alpha$ bounded by the first infinite ordinal ($\seq$), in a sequence that can have any ordinal $\alpha$ as its length ($\ord$), or whether $f=\nonlang$ and therefore the question of an order doesn't make sense ($\na$). Finally, $q\in \{\some, \all\}$ indicates whether at least one agent must learn $\phi$ (\some) or all of them must learn it ($\all$).


We use $f(\tau)$, $a(\tau)$, $o(\tau)$ and $q(\tau)$ to denote the values of $f$, $a$, $o$ and $q$, respectively, in $\tau$.
Each type $\tau$ of distributed knowledge induces semantics for the language $\lang$, which we denote by $\models_\tau$.

The main idea of each of the semantics is that $\phi$ is distributed knowledge among $G$ if there is some way for $G$ to share information among themselves that would result in them learning the truth of $\phi$. Communication by $G$ will change the current information state, which is encoded by the set $\sim$ of relations, into a new information state $\sim'$. Often there are different things that $G$ could communicate, and each such possible communication will lead to a new information state. Hence we will, in general, need to consider not one new information state $\sim'$ but a set of such information states. 

We will denote the set of information states that can be reached by group $G$, when discussing in \rev{world} $s$, using the communication type $\tau$, as $\R{G}{s}{\tau}$.

The semantics for the distributed knowledge operator are then given by $\M,s\models_\tau D_G\phi$ iff
\begin{center}$\exists \sim'\in \R{G}{s}{\tau} \exists a\in G: \M,t\models \phi$ for all $t$ such that  $s\sim'_a t$\end{center}
if $q(\tau)=\some$, or
$\M,s\models_\tau D_G\phi$  iff  
\begin{center}$\exists \sim'\in \R{G}{s}{\tau} \forall a\in G: \M,t\models \phi$ for all $t$ such that  $s\sim'_a t$\end{center}
if $q(\tau)=\all$.

We should note that, while the agents are generally not \emph{required} to share as much information as they can, distributed knowledge is about whether the agents are able to achieve knowledge of $\phi$. The more information is shared, the more likely it is that the agents will learn $\phi$.\footnote{This does rely on the fact that we are looking at \emph{static} communication, i.e. $\phi$ is distributed knowledge if the agents can learn that $\phi$ \emph{used to be} true. In \emph{dynamic} communication, where agents are trying to learn that $\phi$ \emph{is} true, sharing as much information as possible may not be an optimal strategy since $\phi$ can contain ignorance conditions that become false when more information is shared (e.g. the Moore formula).} In particular, if there is a unique ``maximal communication'', it suffices to consider only that communication.

For example, in the intersection definition of distributed knowledge, agents are capable of explaining, in a non-linguistic way, exactly which \rev{worlds} they consider possible. Conceptually, it seems reasonable that agents could, instead of communicating their exact set of possible \rev{worlds}, communicate a superset of it. We need not consider this possibility, however, since sharing the exact set of \rev{worlds} they consider possible is the optimal strategy. In this case, it therefore suffices to consider the singleton set $\R{G}{s}{\tau}=\{\sim'\}$ where, for every $a\in G$, $\sim'_a= \bigcap_{b\in G}\sim_b$.

In other cases there may be no single most informative communication, so we cannot restrict ourselves to a single information state in this way. For example, if every agent can communicate a single formula that is known to them, there is not, in general, a single most informative formula for them to state. For every formula $\psi_G = \bigwedge_{a\in G}\psi_a$ with the property that $\M,s\models \square_a\psi_a$ for every $a$, we therefore need to consider the information state $\sim^{\psi_G}$.

Based on the considerations from the previous section, we can define the following 12 variants of distributed knowledge \rev{(see Figure \ref{fig:expressivity} for the full list of variants and their relative strength)}. Recall, however, that some of these variants are equivalent to one another.



 \paragraph{Non-linguistic sharing} Suppose information is shared non-linguistically. Then our method of restricting the amount of information shared is inapplicable. Furthermore, because we do not know how information is shared we also cannot speak of an ordering in which information is presented.
	
	The only further distinction that is available is whether one agent needs to learn the formula or all of them do. We therefore need to consider the variants $\tau = \nlsingle$ and $\tau = \nlall$, respectively.
	
	As stated above, for either of these cases we have $\R{G}{s}{\alpha}=\{\sim'\}$ where $\sim'_a=\bigcap_{b\in G}\sim_b$. Note, also, that at this point it is already easy to see that $\nlsingle$ and $\nlall$ are equivalent. This is because all agents end up with the same accessibility relation, so if one of them learns $\phi$ then they all do. 

 \begin{proposition}
$\M,s\models_{\nlsingle{}} D_G\phi$ if and only if $\M,s\models_{\nlall{}} D_G\phi$.
 \end{proposition}
 \paragraph{Simultaneous sharing of formulas} Suppose the information sharing happens by the agents stating one or more formulas each, and that this happens simultaneously.
    So we are considering the variants $\tau \in \{\formulasimsingle,$ $\formulasimall,$ $\setsimsingle,$ $\setsimall\}$,
    where the first two assume that each agent contributes a single formula while the last two let each agent contribute a (potentially infinite) set of formulas.

    We specify which information states the agents can achieve by sharing information in two steps. First, we specify all of the ways the agents could share information. In effect, this acts as a set of indices used to identify the various outcome information states. Then, for each index we specify what that outcome information state is.

    The important condition on information sharing is that each agent must contribute one or more formulas that they know.\footnote{Note that we can assume without loss of generality that every agent shares at least one formula because agents can always use the uninformative formula $\top$ that is known by everyone.} Hence if $\tau = \formulasimsingle$ or $\tau = \formulasimall$ then 
    $$\R{G}{s}{\tau}=\{\sim^{\{\psi_a\mid a\in G\}}\mid \forall a\in G: \psi_a\in \lang_0 \text{ and }\M,s\models \square_a\psi_a\}$$
    and
    $$\sim^{\{\psi_a\mid a\in G\}}_\rev{b} = \{(x,y)\in S\times S\mid (x,y)\in \sim_\rev{b} \text{ and } \M,y\models \bigwedge_{a\in G} \psi_a\}.$$

    If each agent shares a set of formulas, i.e. if $\tau = \setsimsingle$ or $\tau = \setsimall$, then they need to know each of the formulas they provide, so
	\begin{align*}
	\R{G}{s}{\tau} = \{\sim^{\{\Psi_a\mid a\in G\}}\mid \forall a\in G:\Psi_a\subseteq \lang_0 \text{ and } \forall \psi \in \Psi_a: \M,s\models \square_a\psi\}
	\end{align*}
	and
	$$\sim^{\{\Psi_a\mid a\in G\}}_\rev{b} = \{(x,y)\in S\times S\mid (x,y)\in \sim_\rev{b} \text{ and } \forall a \forall \psi\in \Psi_a:\M,y\models\psi\}.$$

	
	
 \paragraph{Taking turns} Suppose that, as in the previous case, agents share one or more formulas, but now they do so sequentially. The crucial difference with simultaneous communication is that in the sequential case agents can state a formula that they only know because of the information provided to them by the previous speakers. 
 
    Assume, for example, that $a$ knows $p$ and $b$ knows $p\rightarrow q$. If $a$ speaks first and tells $b$ that $p$ is true, $b$ can then, when it is their turn to speak, say that $q$ is true, which they only know because $a$ told them $p$.

	If every agent shares exactly one formula, then the sequential sharing of information means every agent takes a single turn, where the later agents can use the information provided by the earlier ones.

    If each agent provides a set of formulas, we still need to specify the order among the agents, but this will generally have to be an infinite order. Note that we do not have to assume that this order is ``fair'', since agents can skip their turn by providing the trivial formula $\top$. What is potentially important, however, is whether the turn-taking is limited to $\omega$ rounds (where $\omega$ is the first infinite ordinal), or whether any ordinal can be used. 
    This gives us six variants: $\tau \in \{\formulaseqsingle, \formulaseqall, \setseqsingleomega,$ $\setseqallomega,$ $\setseqsingleOmega,$ $\setseqallOmega\}$. \rev{Observe that we do not consider types $(\yeslang, \single, \Omega, \forall)$ and $(\yeslang, \single, \Omega, \exists)$ since a finite number of agents communicating one formula each will never step on the trans-finite territory. }

    Before we can define the possible effects of sequential communication, we first need a little bit more notation. We want to consider finite sequences, infinite sequences, and even trans-finite sequences of statements. Therefore, let $\alpha$ be any ordinal. At each ordinal $\delta< \alpha$, one of the agents will state the truth of one formula; let us write $f(\delta)$ for the agent and $g(\delta)$ for the formula. As $\alpha$ can be considered to be identical to the set of all ordinals less than 
    it, this means $f$ and $g$ are functions  of type $f:\alpha\rightarrow G$ and $g:\alpha\rightarrow \lang_0$.
	
	Agent $f(\delta)$ needs to know formula $g(\delta)$ at time $\delta$, since otherwise they would not be able to state the truth of the formula. As such, we need to keep track of the information state at each point in the process. Formally, this means that we are interested in the final information state $\sim^{\alpha,f,g}$, but we also need to define $\sim^{\alpha,f,g,\delta}$ for $\delta <\alpha$, which represents the information state immediately after the announcement that takes place at time $\delta$. We do this by defining 
	$$\sim^{\alpha,f,g,0}_a = \{(x,y)\in S\times S \mid (x,y)\in \sim_a \text{ and } \M,y\models g(0)\}$$
	and
	$$\sim^{\alpha,f,g,\delta}_a = \{(x,y)\in S\times S\mid (x,y)\in \bigcap_{\epsilon<\delta}\sim^{\alpha,f,g,\epsilon}_a \text{ and } \M,y\models g(\delta)\}$$
	for $0 < \delta < \alpha$. 
 Finally, we define $\sim^{\alpha,f,g}_a = \bigcap_{\delta<\alpha}\sim_a^{\alpha,f,g,\delta}$.

    Note that $\sim^{\alpha,f,g,\delta}$ is the information state \emph{after} the communication at time $\delta$, and that there is such a communication at every time $\delta < \alpha$. Hence $\sim^{\alpha,f,g,0}_a$ is generally not identical to $\sim_a$, since the latter represents the information state before any communication takes place.

    Furthermore, communication only happens at $\delta < \alpha$. Hence, in particular, if $\alpha = \omega$ then communication takes place at every finite time step, but there is no ``infinity-th'' communication at $\omega$.

    Importantly, the definition of $\sim^{\alpha,f,g}$ does not check whether agent $f(\delta)$ actually knows $g(\delta)$ at time $\delta$.
    So while the definition determines the effect that a given communication sequence would have, it does not determine whether the agents are actually capable of saying the formulas included in the sequence. For this, we need to \rev{look at another property that} we will refer to as \emph{correctness}. 
%
%
%
%
%
%
%
	We say that $\alpha$, $f$ and $g$ are \emph{correct} for group $G$ and \rev{world} $s$ if
	$$\forall s': \text{ if } (s,s')\in \sim_{f(\delta)} \cap \bigcap_{\epsilon < \delta}\sim^{\alpha,f,g,\epsilon}_{f(\delta)}, \text{ then }\M,s'\models g(\delta).$$
	In other words, if $s'$ was accessible originally (for agent $f(\delta)$) and has not been excluded by any of the preceding statements at times $\epsilon < \delta$, then $g(\delta)$ must be true in $s'$.
		
	Now, we can formally define $\R{G}{s}{\tau}$ for the sequential types of communication. If $\tau = \formulaseqsingle$ or $\tau = \formulaseqall$, then we require every agent to have exactly one turn, and hence $\alpha = |G|$. Moreover, $f$ is a bijection, so
	$$\R{G}{s}{\tau} = \{\sim^{|G|,f,g}\mid f \text{ is a bijection and } |G|, f \text{ and } g \text{ are correct for } G, s\}.$$
 
	If $\tau = \setseqsingleomega$ or $\tau = \setseqallomega$, then we take $\alpha = \omega$, i.e., we allow infinite statements, but there is no ``infinity-th statement''. We could demand that the turn-taking by the agents is ``fair'' in some way, but that is pointless; if there is an unfair turn-taking the agents could use for their communication, this can be transformed into a fair one where some agents give the trivial statement $\top$. As such, we get
    $$\R{G}{s}{\tau} = \{\sim^{\omega,f,g}\mid \omega, f \text{ and } g \text{ are correct for } G, s\}.$$
    Finally, we can allow any ordinal number $\alpha$ of statements. If $\tau = \setseqsingleOmega$ or $\tau = \setseqallOmega$, then
    $$\R{G}{s}{\tau} = \{\sim^{\alpha,f,g}\mid \alpha, f \text{ and } g \text{ are correct for } G, s\}.$$
    \rev{Note that all $\R{G}{s}{\tau}$'s are sets, and we essentially quantify over all possible sequences of announcements, and, according to the definition of distributed knowledge, it is enough that at least one knowledge state induced by any sequence satisfies $\varphi$.}

\subsection{Connections to the traditional definitions}
\label{sec:connections}
Our variant $\nlall$ of distributed knowledge is, modulo some notation, identical to the traditional definition of distributed knowledge based on intersection. Our variant that most closely matches the full communication definition of distributed knowledge is $\formulasimall$.

In fact, as mentioned in the introduction, $\formulasimall$ is equivalent to the full communication definition, but this is not entirely obvious and therefore requires a short proof.

\begin{proposition}
\label{prop:connections}
We have $$\M,s\models_\formulasimall D_G\phi$$ if and only if 
$$\{\psi\in \lang_0\mid \exists a\in G: \M,s\models \square_a\psi\}\models \phi.$$
\end{proposition}
\begin{proof}
Suppose $\M,s\models_\formulasimall D_G\phi$. Then there are $\{\psi_\rev{b}\mid \rev{b}\in G\}$ such that (1) for all $\rev{b}\in G$, $\M,s\models \square_\rev{b}\psi_\rev{b}$ and (2) for all $a\in G$ and every $s'$, if $s\sim_a^{\{\psi_\rev{b}\mid \rev{b}\in G\}} s'$ then $\M,s'\models \phi$. Furthermore, $s\sim_a^{\{\psi_\rev{b}\mid \rev{b}\in G\}} s'$ holds if and only if $s\sim_a s'$ and $\M,s'\models \bigwedge \{\psi_\rev{b}\mid \rev{b}\in G\}$.

This implies that for every $s'$, if $s\sim_a s'$, then $\M,s'\models \bigwedge\{\psi_\rev{b}\mid \rev{b}\in G\}\rightarrow \phi$. Furthermore, since $\M,s\models \square_a\psi_a$, we also have $\M,s'\models \psi_a$. It follows that $\M,s\models \square_a(\psi_a\wedge (\bigwedge\{\psi_\rev{b}\mid \rev{b}\in G\}\rightarrow \phi))$.

Now, note that $\{\psi_a \wedge (\bigwedge\{\psi_\rev{b}\mid \rev{b}\in G\}\rightarrow \phi)\mid a\in G\}\models \phi$. As such, $\{\psi \mid \exists a\in G: \M,s\models \square_a\psi\}\models \phi$.

For the other direction, suppose that $\{\psi \mid \exists a\in G: \M,s\models \square_a\psi\}\models \phi$. Since epistemic logic is compact, there is a finite $\Psi\subseteq\{\psi \mid \exists a\in G: \M,s\models \square_a\psi\}$ such that $\Psi\models \phi$. For every $a\in G$, let $\psi_a = \bigwedge \{\psi\in \Psi\mid \M,s\models \square_a\psi\}$.

We now have $\M,s\models \square_a\psi_a$ for every $a\in G$. Furthermore, since $\Psi\models \phi$, we also have $\M,s'\models \bigwedge\{\psi_a\mid a\in G\}\rightarrow \phi$ \rev{for every $s'\in S$ such that $s\sim_as'$}. It follows that $\M,s\models_\rev{\formulasimall} D_G\phi$.
\end{proof}
Note that the proof critically depends on the compactness of epistemic logic. If we used a non-compact base logic, such as epistemic logic with common knowledge, the equivalence would not hold.

\section{Relative Strength}
\label{sec:expressivity}
Now that we have formally defined the various types of distributed knowledge that we are interested in, we can investigate their properties. In particular, we are interested in which variants imply each other.


In most cases, it is clear that one variant $\tau_1$ is at least as strong as another variant $\tau_2$, in the sense that $\M,s\models_{\tau_1}D_G\phi$ implies $\M,s\models_{\tau_2}D_G\phi$.
In particular, it is easy to see that the following hold.
\begin{proposition}\hfill
\label{prop:basic_comparisons}
\begin{itemize}
    \item For every $\tau$, $\M,s\models_\tau D_G\phi$ implies $\M,s\models_\nlall \rev{D_G}\phi$.
    
    \item For every $f$, $a$ and $o$, $\M,s\models_{(f,a,o,\all)}D_G\phi$ implies $\M,s\models_{(f,a,o,\some)}D_G\phi$.
    \item For all $o$ and $q$, $\M,s\models_{(\yeslang,\single,o,q)}D_G\phi$ implies $\M,s\models_{(\yeslang,\multi,o,q)}D_G\phi$.
    
    \item For all $a$ and $q$, $\M,s\models_{(\yeslang,a,\simult,q)}D_G\phi$ implies $\M,s\models_{(\yeslang,a,\seq,q)}D_G\phi$.

    Furthermore, $\M,s\models_{(\yeslang,\multi,\seq,q)}D_G\phi$ implies $\M,s\models_{(\yeslang,\multi,\ord,q)}D_G\phi$.
\end{itemize}
\end{proposition}
Furthermore, it follows from \cite{vanderhoek99} that $\nlall$ and $\nlsingle$ do not imply $\formulasimall$.
This, however, leaves many comparisons open, which we solve here.


\subsection{Single formula}
\label{subsec:single}
First, let us compare variants that use a single formula. In particular, we will show that all the four variants, i.e. $\formulasimsingle$, $\formulasimall$, $\formulaseqsingle$, and $\formulaseqall$, are equivalent

\begin{proposition}
\label{prop:single}
$\M,s\models_\formulasimsingle D_G\phi$ 
if and only if $\M,s\models_\formulasimall D_G\phi$.
\end{proposition}
\begin{proof}
We know from Proposition~\ref{prop:basic_comparisons} that ``all'' implies ``single'', so it suffices to show that the reverse also holds.
Suppose, therefore, that $\M,s\models_{\formulasimsingle} D_G\phi$. Let $\{\psi_a\mid a\in G\}$ be the witnessing formulas that, if communicated among the group, would make one agent, let's call them $x$, learn that $\phi$ is true.

This means that for every $s'$, if $(s,s')\in \sim_x$ and $\M,s'\models \bigwedge_{a\in G}\psi_a$, then $\M,s'\models \phi$. This implies that for every $(s,s')\in \sim_x$, we have $\M,s'\models \bigwedge_{a\in G}\psi_a\rightarrow \phi$. Furthermore, since $x$ was able to provide the formula $\psi_x$, we also have $\M,s'\models \psi_x$ for each such $s'$. Hence $\M,s\models \square_x(\psi_x \wedge (\bigwedge_{a\in G}\psi_a\rightarrow \phi))$.

Consider then the alternative communication $\{\psi'_a\mid a\in G\}$ where $\psi'_a = \psi_a$ for $a\not = x$ and $\psi'_x = \psi_x \wedge (\bigwedge_{a\in G}\psi_a\rightarrow \phi)$.
For every agent $a\in G$ and every $s'$, if $(s,s')\in \sim_a$ and $\M,s'\models \bigwedge_{a\in G}\psi_a'$, we then have, in particular, $\M,s'\models \bigwedge_{a\in G}\psi_a$ and $\M,s'\models \bigwedge_{a\in G}\psi_a\rightarrow \phi$, and hence $\M,s'\models \phi$. This implies that all agents learn $\phi$, and therefore $\M,s\models_\formulasimall D_G\phi$.
\end{proof}
In effect, the single agent $x$ that learns $\phi$ can include hypothetical reasoning in the formula that they provide to the group. Instead of saying ``$\psi_x$ is true'', they can say ``$\psi_x$ is true, and if you were to tell me $\{\psi_a\mid a\in G\}$, then I would learn that $\phi$ is true''. This suffices for all the other agents to learn the truth of $\phi$, if all $a\not = x$ do indeed provide formulas $\psi_a$.
The same trick can be used in the sequential version.
\begin{proposition}
\label{prop:finite_singleall}
$\M,s\models_\formulaseqsingle D_G\phi$ 
if and only if $\M,s\models_\formulaseqall D_G\phi$.
\end{proposition}
\begin{proof}
\rev{Let $(a_1,\cdots,a_n)$ be an ordering of $G$} such that $a_i$ takes their turn before $a_j$ iff $i<j$, and let \rev{$(\psi_1,\cdots, \psi_n)$} be the corresponding formulas that witness $\M,s\models_\formulaseqsingle D_G\phi$, where $a_x$ is the agent that learns $\phi$.

Then $a_x$ already knows, before the communication starts, that $\bigwedge \{\psi_i \mid 1\leq i \leq n\} \rightarrow \phi$. Furthermore, once it is their turn, they have also learned that $\psi_{a_x}$. Hence they could instead say $\psi'_{a_x}=\psi_{a_x} \wedge (\bigwedge \{\psi_i \mid 1\leq i \leq n\} \rightarrow \phi)$, which would result in all agents learning $\phi$. \rev{Note that since we work with the static notion of distributed knowledge, rather than a dynamic one, knowledge of agents is monotonic under announcements and they can always announce their respective formulas.}
\end{proof}

Additionally, a similar kind of hypothetical reasoning can be used to remove the reliance on sequential communication.

\begin{proposition}
\label{prop:finite_sequence}
$\M,s\models_\formulaseqall D_G\phi$ 
if and only if $\M,s\models_\formulasimall D_G\phi$.
\end{proposition}
\begin{proof}
\rev{By Proposition~\ref{prop:basic_comparisons}, $\M,s\models_\formulasimall D_G\phi$ implies $\M,s\models_\formulaseqall D_G\phi$. Left to show is the other direction.}

Suppose \rev{therefore} that $\M,s\models_\formulaseqall D_G\phi$, as witnessed by order \rev{$(a_1,\cdots,a_n)$} and formulas \rev{$(\psi_1,\cdots,\psi_n)$}. Then, at stage $i$, the preceding communications $\{\psi_j \mid j<i\}$ suffice for agent $i$ to learn that $\psi_i$ holds, in the sense that for all $s'$, if $(s,s')\in \sim_{a_i}$ and $\M,s'\models \bigwedge_{j<i}\psi_j$ then $\M,s'\models \psi_i$.

It follows that, before the communication started, $a_i$ already knew $\bigwedge_{j<i}\psi_j \rightarrow \psi_i$. So, in the simultaneous version, $a_i$ could provide that formula.

Furthermore, collectively, communicating $\{\psi_i \mid 1\leq i \leq n\}$ has the same effect as communicating $\{\bigwedge_{j<i}\psi_j \rightarrow \psi_i \mid 1\leq i \leq n\}$. So we also have $\M,s\models_\formulasimall D_G\phi$.
\end{proof}
The equivalence of other pairs of the single formula variants follows immediately from the transitivity of the equivalence relation.

\subsection{Sets of formulas}
Let us now consider the variants where each agent may provide a set of formulas. As a first step, we will show that none of them imply the single formula variants. For this, it suffices to show that any one of the single formula variants is not implied by the strongest set variant (i.e., the set variant that is the hardest to satisfy, which is \setsimall).

\begin{restatable}{proposition}{propFormulaSet}
\label{prop:formula_set}
$\M,s\models_\setsimall D_G\phi$ does not necessarily imply that $\M,s\models_\formulasimall D_G\phi$.
\end{restatable}
\begin{proof} (\textit{Sketch; full proof in Appendix~\ref{appendix:proofs}}).
Suppose that, for every $i\in \mathbb{N}$, agent $a$ knows whether $p_i$ holds while $b$ knows whether $q_i$ holds. Furthermore, suppose both agents know that $r$ is true if and only if, for every $i$,  $p_i\leftrightarrow q_i$ holds.
Consider the case where all $p_i$ and $q_i$ are, in fact, false, so $r$ holds.

Under the $\setsimall$ definition, we then have $D_{\{a,b\}}r$. After all, $a$ can tell $b$ that all $p_i$ are false, while $b$ can tell $a$ that all $q_i$ are false. This suffices for both of them to discover that $r$ is true.

Under the $\formulasimall$ definition, however, we have $\neg D_{\{a,b\}}r$. This is because the agents can only learn that $r$ is true if $p_i\leftrightarrow q_i$ for all $i$, which cannot be expressed in a finite set of formulas.
\end{proof}



Next, let us note that $\setsimsingle$ does not imply $\setsimall$.
\begin{restatable}{proposition}{propsomeall}
$\M,s\models_\setsimsingle D_G\phi$ does not necessarily imply $\M,s\models_\setsimall D_G\phi$.
\end{restatable}
\begin{proof} (\textit{Sketch; full proof in Appendix~\ref{appendix:proofs}}).
As in Proposition~\ref{prop:formula_set}, suppose $a$ knows whether $p_i$ is true and $b$ knows whether $q_i$ is true. Now, however, suppose that $b$ knows that whether $r$ is true depends on the parity of the number of indices $i$ such that $p_i$ and $q_i$ differ in value. Specifically, $b$ knows that $r$ is true if that number is even, while $a$ is uncertain whether $r$ holds if the number i\rev{s} even, or if it is odd. \rev{Both agents know $r$ is false if $p_i$ and $q_i$ differ infinitely often.} As before, suppose that all $p_i$ and $q_i$ happen to be false.

With the $\setsimsingle$ definition of distributed knowledge, we then have $D_{\{a,b\}}r$. This is because, when $a$ tells $b$ that all $p_i$ are false, agent $b$ will learn that there are 0 indices where $p_i$ and $q_i$ differ, so $r$ is true.

Yet $r$ is not distributed knowledge under the $\setsimall$ definition of distributed knowledge, since only $b$ can learn that $r$ is true. The reason $a$ can't learn this is that ``$r$ is true iff there is an even number of $i$ such that $p_i$ and $q_i$ disagree'' cannot be expressed in epistemic logic. Furthermore, while $b$ learns that $r$ is true once the communication is complete, $\setsimall$ requires simultaneous communication, so $b$ cannot simply say that $r$ is true.
\end{proof}



What does not make a difference, however, is simultaneous stat\rev{e}ments or $\omega$-sequential ones.
\begin{proposition}
$\M,s\models_\setseqsingleomega D_G\phi$ iff $\M,s\models_\setsimsingle D_G\phi$, and $\M,s\models_\setseqallomega D_G\phi$ iff $\M,s\models_\setsimall D_G\phi$
\end{proposition}
\begin{proof}
At every step in the $\omega$-sequential communication, when $\psi_i$ is stated, a finite set $\{\psi_j \mid j<i\}$ preceded it. We can 
replace sequential announcement of $\{\psi_i\mid i \in \mathbb{N}\}$ by simultaneous announcement of $\{\bigwedge_{j<i}\psi_j\rightarrow \psi_i\mid i\in \mathbb{N}\}$.
\end{proof}

The $\Omega$-sequential variant, on the other hand, is strictly weaker than $\omega$-sequential or simultaneous ones.

\begin{restatable}{proposition}{propomegaOmega}
$\M,s\models_\setseqsingleOmega D_G\phi$ does not necessarily imply $\M,s\models_\setsimsingle D_G\phi$.
\end{restatable}
\begin{proof} (\textit{Sketch; full proof in Appendix~\ref{appendix:proofs}}).
Suppose that $a$ knows, for all $i$ and $j$, whether $p_{i,j}$ and $q_{i,j}$ hold. Furthermore, suppose that the value of $x_i$ depends on the number of $j$ such that $p_{i,j}$ differs from $q_{i,j}$, in a way known to $b$ but not to $a$ and $c$, \rev{and that this dependence cannot be expressed in epistemic logic. (See the full proof in the appendix for one way to create such an inexpressible dependence.) Similarly, $y_i$ depends on the number of $j$ such that $p_{i,j}$ and $q_{i,j}$ differ, in an inexpressible way that is known to $c$ but not to $a$ and $b$.}
Finally, all three agents know that $z$ is true iff there is an even number of $i$ such that $x_i$ and $y_i$ differ.

Then if $z$ is true, that is distributed knowledge using the $\setseqsingleOmega$ definition, since $a$ can tell $b$ and $c$ which $p_{i,j}$ and $q_{i,j}$ \rev{hold}, at which point they can say which $x_i$ and $y_i$ are true, allowing all of them to determine that $z$ holds.

With the \rev{$\setsimsingle$} definition $z$ is not distributed knowledge, however. 
\rev{This is because, in order for any of the three agents to learn that $z$ is true, all three agents need to contribute their information. But $b$ and $c$ cannot initially contribute any non-trivial formulas, since their only private information is the way in which $x_i$ or $y_j$ depends on the values of $p_{i,j}$ and $q_{i,j}$, and this dependence is not expressible in epistemic logic. 

It is only after $a$ has informed $b$ and $c$ about the values of every $p_{i,j}$ and $q_{i,j}$ that $b$ and $c$ can apply their knowledge in order to determine the truth of $x_i$ and $y_j$, respectively, which they can then communicate to the other agents. So $a$ first needs to contribute at least $\omega$ formulas before $b$ and $c$ can get involved, which is not possible if \rev{$\tau = \setsimsingle$.}}
\end{proof}

However, if we allow any ordinal number of communication steps, the difference between a single agent learning the formula or all of them doing so disappears.

\begin{proposition}
$\M,s\models_\setseqsingleOmega D_G\phi$ 
if and only if $\M,s\models_\setseqallOmega D_G\phi$.
\end{proposition}
\begin{proof}
Suppose that $\M,s\models_\setseqsingleOmega D_G\phi$. So there is a sequence of communication\rev{s} of length $\alpha$ after which one agent $a$ knows $\phi$. Now, consider the sequence of length $\alpha + 1$ where, in the last step, agent $a$ states that $\phi$ is true. This suffices for all agents to learn that $\phi$ is true.
\end{proof}
The above suffices to determine the comparative strength of each of the variants we discussed.

\section{Discussion}
\label{sec:conclusion}
We have considered 12 natural interpretations of the 
idea behind distributed knowledge. For these interpretations, we have analysed  which ones are equivalent to each other, and which ones are different. The complete landscape of distributed knowledge is shown in Figure \ref{fig:expressivity}.
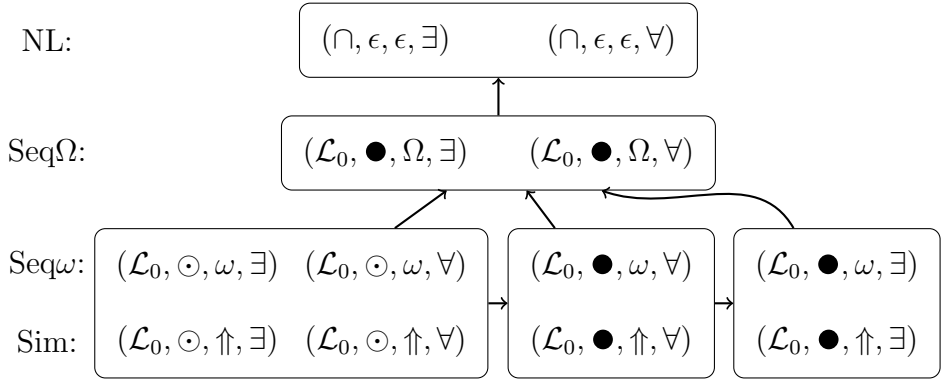
\begin{figure}[h!]
\centering
\rev{\begin{tikzpicture}[rounded corners,baseline,every node/.style={transform shape}]
\node (classic1) at (2.5,1) {$\nlsingle$};
\node (classicALL) at (5.5,1) {$\nlall$};

\node (seqOmega1) at (2.5,-0.5) {$\setseqsingleOmega$};
\node (seqOmegaALL) at (5.5,-0.5) {$\setseqallOmega$};

\node (seqomegasingle1) at (0, -2) {$\formulaseqsingle$};
\node (seqomegasingleALL) at (2.5,-2) {$\formulaseqall$};
\node (seqomegaset1) at (5.5,-2) {$\setseqallomega$};
\node (seqomegasetALL) at (8.5,-2) {$\setseqsingleomega$};

\node (simsingle1) at (0, -3) {$\formulasimsingle$};
\node (simsingleALL) at (2.5,-3) {$\formulasimall$};
\node (simsetALL) at (5.5,-3) {$\setsimall$};
\node (simset1) at (8.5,-3) {$\setsimsingle$};


\node (NL) [draw,  fit= (classic1) (classicALL)] {};
\node (seqO) [draw,  fit= (seqOmega1) (seqOmegaALL)] {};
\node (single) [draw,  fit= (seqomegasingle1) (simsingleALL)] {};
\node (setone) [draw,  fit= (seqomegasetALL) (simset1)] {};
\node (setall) [draw,  fit= (seqomegaset1) (simsetALL)] {};

\draw[thick,->] (single) -- (setall);
\draw[thick,->] (setall) -- (setone);
\draw[thick,->] (single) -- (seqO);
\draw[thick,->] (setone) to [out = 120, in = -20] (seqO);
\draw[thick,->] (setall) -- (seqO);
\draw[thick,->] (seqO) -- (NL);

\node (nl) at (-2, 1) {NL:};
\node (seqOmega) at (-2, -0.5) {Seq$\Omega$:};
\node (seqomega) at (-2, -2) {Seq$\omega$:};
\node (sim) at (-2, -3) {Sim:};
\end{tikzpicture}}
\caption{The expressivity landscape of distributed knowledge, where `NL' stands for `non-linguistic sharing', `Seq$\Omega$' denotes unlimited sequential sharing,  `Seq$\omega$' stands for sequential sharing limited to ordinal $\omega$, and `Sim' denotes simultaneous sharing.
\rev{Equivalent variations of distributed knowledge are enclosed in a box. Arrows point from stronger variants to weaker ones. Some arrows that follow from transitivity have been omitted for the sake of clarity.}
}
\label{fig:expressivity}
\end{figure}

Out of all variants of distributed knowledge, only the classic one (i.e. $\nlall$ and $\nlsingle$ in our taxonomy) was axiomatised \cite{fagin92,vanderhoek92}. The task of providing axiomatisations for the remaining variants seems to be both colossal and irresistibly tempting, and we thus leave it for future work.  

\rev{\subsection*{Acknowledgements}
A preliminary version of this work has benefited from the discussion with Fernando R. Vel{\'a}zquez-Quesada. We would also like to thank the anonymous reviewers of AiML for their attention to details and constructive criticism that helped us to improve the presentation of this work.}


\appendix
\section{Vicious Circularity}
\label{app:vicious}
Suppose that we define the type $(\lang,\single,\simult,\all)$ in the same way as the type $(\yeslang,\single,\simult,\all)$, except that we now allow $a\in G$ to contribute any formula $\psi_a\in \lang$ that they know, as opposed to any known formula from $\lang_0$.

Clearly, such a definition would be circular. Here, we show that this circularity is vicious. More specifically, it is under-determined, in the sense that there are models where both $D_G\phi$ and $\neg D_G\phi$ are consistent with the semantics.

\begin{proposition}
There are $\M,s$ and $\phi$ such that both $\M,s\models_{(\lang,\single,\simult,\all)}D_G\phi$ and $\M,s\models_{(\lang,\single,\simult,\all)}\neg D_G\phi$ are consistent with the $(\lang,\single,\simult,\all)$-semantics.
\end{proposition}
\begin{proof}
Consider the following model:
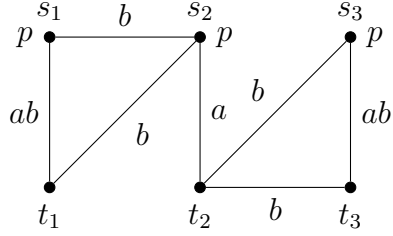
\begin{center}\begin{tikzpicture}
\node[circle,fill=black,draw=black, minimum size=4pt,inner sep=0pt, label=left:{$p$}, label=above:{$s_1$}](s1) at (0,0) {};
\node[circle,fill=black,draw=black, minimum size=4pt,inner sep=0pt, label=right:{$p$}, label=above:{$s_2$}](s2) at (2,0) {};
\node[circle,fill=black,draw=black, minimum size=4pt,inner sep=0pt, label=right:{$p$}, label=above:{$s_3$}](s3) at (4,0) {};
\node[circle,fill=black,draw=black, minimum size=4pt,inner sep=0pt, label=below:{$t_1$}](t1) at (0,-2) {};
\node[circle,fill=black,draw=black, minimum size=4pt,inner sep=0pt, label=below:{$t_2$}](t2) at (2,-2) {};
\node[circle,fill=black,draw=black, minimum size=4pt,inner sep=0pt, label=below:{$t_3$}](t3) at (4,-2) {};

\draw (s1) -- (t1) node[midway,left] {$ab$};
\draw (s2) -- (t2) node[midway,right] {$a$};
\draw (s3) -- (t3) node[midway,right] {$ab$};

\draw (s1) -- (s2) node[midway,above] {$b$};
\draw (t1) -- (s2) node[midway,below right] {$b$};

\draw (s3) -- (t2) node[midway,above left] {$b$};
\draw (t2) -- (t3) node[midway,below] {$b$};
\end{tikzpicture}\end{center}
We will show that both $\M,s_2\models \dist_{\{a,b\}}p$ and $\M,s_2\not\models \dist_{\{a,b\}}p$ are consistent with the circular semantics. To see why this is the case, first note that the semantics are extensional, so while there are infinitely many formulas it suffices to consider only those that have different extensions. Let use denote the extension of $\phi$ by $\true{\phi}$.





Let $E$ be the set of all extension on this model. In order for $E$ to be consistent with the $(\lang,\single,\simult,\all)$-semantics, it has to be ``self-fulfilling'', in the sense that, if we assume that $E$ is the set of all extensions, then we should have $\true{\phi}\in E$ for all $\phi\in \lang$, and for every $e\in E$ there should be some $\phi_e\in \lang$ such that $\true{\phi_e}=e$.

We will show that there are two different sets of extensions that satisfy this criterion: we can take $E=2^S$, in which case we have $\M,s_2\models D_{\{a,b\}}p$, and we can take $E=\{\emptyset, \{s_1,s_2,s_3\},\{t_1,t_2,t_3\},S\}$, in which case $\M,s_2\not\models D_{\{a,b\}}p$.

Suppose therefore that $E=2^S$. It is immediate that $\true{\phi}\in E$ for all $\phi\in \lang$. Left to show is that every extension $e\in E$ is witnessed by some formula $\phi_e$. To this purpose, we first note that we have $\M,s_2\models D_{\{a,b\}}p$. This is because, by assumption, $E=2^S$ is the set of extensions, so there are formulas $\phi_1$ and $\phi_2$ such that $\true{\phi_1}=\{s_2,t_2\}$ and $\true{\phi_2}= \{s_1,s_2,t_1\}$. Then $\M,s_2\models \square_a\phi_1$ and $\M,s_2\models \square_b\phi_2$, so the agents can share $\phi_1$ and $\phi_2$. Furthermore, by putting $\phi_1$ and $\phi_2$ together, the agents discover that $s_2$ is the only possible world. As $p$ is true there, we have $\M,s_2\models D_{\{a,b\}}p$.

In every other \rev{world}, $D_{\{a,b\}}p$ is false. For $t_1, t_2$ and $t_3$ this follows from the fact that distributed knowledge is truthful and $p$ is false in $t_1, t_2$ and $t_3$. For $s_1$ and $s_3$ it follows from the fact that there is an $ab$-successor ($t_1$ or $t_3$, respectively) where $p$ is false. Since this \rev{world} is an $ab$-successor, it can never be excluded by any formula known to $a$ or $b$, so the agents cannot exclude the possibility of $\neg p$ by combining their information.

We have now shown that the formula $D_{\{a,b\}}p$ uniquely identifies the \rev{world} $s_2$. 
It follows that there are also formulas uniquely identifying every other \rev{world}:
\begin{description}
	\item[$s_1$:] $p\wedge \neg \dist_{\{a,b\}}p\wedge \lozenge_b\dist_{\{a,b\}}p$
	\item[$s_3$:] $p\wedge \neg \dist_{\{a,b\}}p\wedge \neg\lozenge_b\dist_{\{a,b\}}p$
	\item[$t_1$:] $\neg p \wedge \lozenge_b\dist_{\{a,b\}}p$
	\item[$t_2$:] $\neg p \wedge \lozenge_a\dist_{\{a,b\}}p$
	\item[$t_3$:] $\neg p \wedge \neg \lozenge_b\dist_{\{a,b\}}p\wedge \neg \lozenge_a\dist_{\{a,b\}}p$.
\end{description}
Let us denote the formula for any \rev{world} by $\phi_{s_i}$ or $\phi_{t_i}$. Any $e\in E$ is the extension of some disjunction of the relevant $\phi_{s_i}$ and/or $\phi_{t_i}$.

We have now shown that $E=2^S$ being the set of extensions is consistent with the semantics, and that we then have $\M,s_2\models D_{\{a,b\}}p$. The witnessing formulas for this distributed knowledge had to have extensions $\{s_2,t_2\}$ and $\{s_1,t_1,s_2\}$, so we can take $\phi_a = \{\phi_{s_2}\vee\phi_{t_2}\}$ and $\phi_b = \{\phi_{s_1}\vee \phi_{s_2}\vee \phi_{t_1}\}$.


Next, we will show 
that it is consistent with the semantics to have $E=\{\emptyset,$ $\{s_1,s_2,s_3\},$ $\{t_1,t_2,t_3\},$ $S\}$, in which case $\M,s_2\not\models D_{\{a,b\}}p$. In this case, it is easy to see that for every $e\in E$ there is some $\phi$ such that $e=\true{\phi}$; we have $\emptyset = \true{\bot}$, $\{s_1,s_2,s_3\}=\true{p}$, $\{t_1,t_2,t_3\}=\true{\neg p}$ and $S=\true{\top}$.

Left to show, therefore, is that for every $\phi$, we have $\true{\phi}\in E$.
Because the $s_i$ are bisimilar to each other, as are the $t_j$, the bisimulation-invariance of modal logic implies that every $\psi\in\lang_0$ will have one of the four extentions in $E$. 

Furthermore, the $D_G$ operator cannot break this symmetry. This is because, in every \rev{world}, both $a$ and $b$ consider at least one $t_i$ \rev{world} and at least one $s_i$ \rev{world} possible. It follows that, for every \rev{world} $x$, if $\M,x\models \square_a\phi$ or $\M,x\models \square_b\phi$, then $\true{\phi}$ contains at least one $t_i$ and at least one $s_i$. Among the four extensions in $E$, the only one with this property is $S$. 

In their communication, $a$ nor $b$ can therefore only contribute formulas with extension $S$, so neither of them provides non-trivial information. It follows that $\phi$ is distributed knowledge if and only if one of the agents already knew $\phi$ before the agents started sharing information. Hence $\M,x\models D_{\{a,b\}}\phi$ if and only if $\M,x\models \square_a\phi \vee \square_b\phi$. It follows that, for every $\phi\in \lang$ there is a $\phi_0\in \lang_0$ such that $\true{\phi}=\true{\phi_0}$, which implies that $\true{\phi}\in E$.

Furthermore, because $\M,s_2\not\models \square_ap \vee \square_bp$, we have $\M,s_2\not\models D_{\{a,b\}}p$.
\end{proof}


\section{Full versions of proofs}
\label{appendix:proofs}
\propFormulaSet*
\begin{proof}
Let $\M=(S,\sim,V)$, where $S$ and $\sim$ are given as follows.
\begin{itemize}
	\item $S= \{s_{x,y}\mid (x,y)\in [0,1)\times [0,1)\}$,
	\item $\sim_a = \{(s_{x,y},s_{x,y'})\}$
	\item $\sim_b = \{(s_{x,y},s_{x',y})\}$
\end{itemize}
Each $x$ or $y$ can be interpreted, written in binary notation, as a set of natural numbers. Let us write $\ats{x}$ for the set of natural numbers represented by $x$. Then we take $V(p_i) = \{s_{x,y}\mid i\in \ats{x}\}$, $V(q_i) = \{s_{x,y}\mid i\in \ats{y}\}$ and $V(r) = \{s_{x,y}\mid x=y\}$.

In every \rev{world}, $a$ knows the $x$-coordinate while being uncertain about the $y$-coordinate, while $b$ is uncertain about the $x$-coordinate while knowing the $y$-coordinate. Since the value of $p_i$ depends only on the $x$-coordinate, this means that $a$ knows, for every $i$, whether $p_i$ is true. Similarly, $b$ knows whether $q_i$ is true.

Furthermore, it is true throughout the model, and therefore known to both agents, that $r$ holds if and only if $x=y$, and therefore if and only if $p_i \leftrightarrow q_i$ for all $i$.

We have $\M, s_{0,0}\models_\setsimall D_{\{a,b\}}r$, which is witnessed by the sets $\Psi_a = \{\neg p_i\mid i\in \mathbb{N}\}$ and $\Psi_b = \{\neg q_i\mid i\in \mathbb{N}\}$. After all, the only \rev{world} where all $p_i$ and $q_i$ are false is $s_{0,0}$, where $r$ is true.

Suppose now, towards a contradiction, that there are formulas $\psi_a$ and $\psi_b$ that are known by their respective agents in $s_{0,0}$, and that would allow the agents to learn that $r$ is true. Let $Q$ be the set of atoms that occur in $\psi_a$ and $\psi_b$. Note that this is a finite set.

Because $\psi_a$ is known by $a$, it must hold in every $s_{0,y}$. Similarly, $\psi_b$ holds on every $s_{x,0}$. Let $x$ and $y$ be such that $x\not = 0$, $y\not = 0$, $\ats{x}\cap Q = \ats{y}\cap Q = \emptyset$ and $x\not = y$. Now, consider the \rev{world} $s_{x,y}$. We will show that it is $Q$-bisimilar to both $s_{x,0}$ and $s_{0,y}$.

To this end, consider the relation $\approx \subseteq S\times S$, such that $s_{x,y} \approx s_{x',y'}$ iff $s_{x,y}$ and $s_{x',y'}$ agree on $Q\cup \{r\}$. We claim that this relation is a bisimulation. Atomic agreement (when restricted to $Q$) between any two $\approx$-related \rev{worlds} is immediate from the construction. We show forth for $b$, the other cases can be shown similarly.

So suppose that $s_{x,y}\approx s_{x',y'}$ and that $s_{x,y}\sim_b s_{u,v}$. Then $y=v$. Now, let $v' = y'$ and
\begin{itemize}
    \item if $u=v$ then $u'=v'$,
    \item if $u\not = v$ then $u'$ is any number such that  $u'\not = v'$ and $\ats{u'}\cap Q = \ats{u}\cap Q$.
\end{itemize}
Then $s_{u,v}$ and $s_{u',v'}$ agree on all atoms in $Q\cup \{r\}$. Hence $s_{u,v}\approx s_{u',v'}$ by our assumption. Furthermore, because $y'=v'$, we also have $s_{x',y'}\sim_b s_{u',v'}$ by the construction of the model. So the forth condition is satisfied.


From this bisimilarity and the fact that $\M,s_{x,0}\models \psi_b$, it follows  that $\M,s_{x,y}\models \psi_b$, and therefore also that $\M,s_{0,y}\models \psi_b$. For the same reason, $\M,s_{x,0}\models \psi_a$.

We can thus conclude that $\psi_a$ and $\psi_b$ cannot exclude the \rev{worlds} $s_{x,0}$ and $s_{0,y}$. In both these \rev{worlds} $r$ is false, so neither $a$ nor $b$ learns that $r$ is true. This contradicts our assumptions, so we have that $\M,s\not\models_\formulasimall D_{\{a,b\}}r$.
\end{proof}

\propsomeall*
\begin{proof}
Let 
\begin{itemize}
    \item $S= \{s_{x,y}\mid (x,y)\in [0,1)\times [0,1)\}\cup \{t_{x,y}\mid (x,y)\in [0,1)\times [0,1)\}$,
    \item $\sim_a = \{(u_{x,y},v_{x,y'})\mid u,v\in \{s,t\}, x,y,y'\in [0,1)\}$
    \item $\sim_b = \{(s_{x,y},s_{x',y})\mid x,x',y\in [0,1)\}\cup\{(t_{x,y},t_{x',y})\mid x,x',y\in [0,1)\}$
    \item $V(p_i) = \{s_{x,y}\mid i \in \ats{x}\}\cup \{t_{x,y}\mid i \in \ats{x}\}$
    \item $V(q_i) = \{s_{x,y}\mid i \in \ats{y}\}\cup \{t_{x,y}\mid i \in \ats{y}\}$
    \item $V(r) = \{s_{x,y}\mid |\ats{x}\triangle \ats{y}| \text{ is even}\}\cup \{t_{x,y}\mid |\ats{x}\triangle \ats{y}| \text{ is odd}\}$\footnote{For sets $X$ and $Y$, the symmetric difference $X \triangle Y$ is defined as $(X\setminus Y) \cup (Y \setminus X)$.}
\end{itemize}
Essentially, the model consists of two grids, one with $s_{x,y}$ \rev{worlds} and another one with $t_{x,y}$ \rev{worlds}. Agent $a$ does not know which grid the current \rev{world} belongs to, while agent $b$ does. On $s_{x,y}$ \rev{worlds} $r$ is true if the number atoms on which $x$ and $y$ differs is even, whereas on $t_{x,y}$ \rev{worlds} $r$ is true if that number is odd. As before, $a$ knows the value of the $p_i$ atoms, while $b$ knows the value of $q_i$. Finally, $b$ can tell the difference between the $s$ and $t$ \rev{worlds}, and therefore the required parity, while $a$ does not.

We have $\M,s_{0,0}\models_\setsimsingle D_{\{a,b\}}r$ because, as in the previous proof, $a$ can announce which $p_i$ hold and $b$ can announce which $q_i$ hold, which suffices for $b$ to determine that $r$ holds.

There is no way for $a$ to learn that $r$ is true by simultaneous statements, however. This follows from another $Q$-bisimilarity argument.
\end{proof}

\propomegaOmega*
\begin{proof}
Let 
\begin{align*}S = \{s_{e,f,g}\mid {} & {} 
 e:\mathbb{N}\times \mathbb{N}\times \{0,1\}\rightarrow \{0,1\}\\
&f: \mathbb{N}\rightarrow \mathbb{N}\\
&g: \mathbb{N}\rightarrow \mathbb{N}\}
\end{align*}
and
\begin{itemize}
    \item $s_{e,f,g}\sim_a s_{e',f',g'}$ iff $e=e'$,
    \item $s_{e,f,g}\sim_b s_{e',f',g'}$ iff $f = f'$,
    \item $s_{e,f,g}\sim_c s_{e',f',g'}$ iff $g = g'$.
\end{itemize}
In other words, in $s_{e,f,g}$ agent $a$ knows $e$, agent $b$ knows $f$ and agent $c$ knows $g$.

Furthermore, for $i,j\in \mathbb{N}$ let 
\begin{itemize}
    \item $V(p_{i,j}) = \{s_{e,f,g}\mid e(i,j,0)=1\}$, 
    \item $V(q_{i,j}) = \{s_{e,f,g}\mid e(i,j,1)=1\}$,
    \item $V(x_{j}) = \{s_{e,f,g}\mid$ the number of indices $i\in \mathbb{N}$ s.t.\ $e(i,j,0)\not = e(i,j,1)$ is divisible by $f(j)\}$
    \item $V(y_{j}) = \{s_{e,f,g}\mid$ the number of indices $i\in \mathbb{N}$ s.t.\ $e(i,j,0)\not = e(i,j,1)$ is divisible by $g(j)\}$
    \item $V(z) = \{s_{e,f,g}\mid$ there is an even number of indices $j\in \mathbb{N}$ s.t.\ exactly one of $x_j, y_j$ holds on $s_{e,f,g} \}$
\end{itemize}

So $e$ simply determines which $p_{i,j}$ and $q_{i,j}$ hold. The function $f$, meanwhile, determines for each $j$ the number, $f(j)$, that must divide the amount of indices $i$ on which $p_{i,j}$ and $q_{i,j}$ are different, in order for $x_j$ to be true. Similarly, $g$ determines the number of indices that must be different for $y_j$ to be true. Finally, $z$ holds if and only if $x_j$ and $y_j$ differ in value for an even number of indices $j$.

It is now easy to see that for any $s_{e,f,g}$ such that $\M,s_{e,f,g}\models z$ we have $\M,s_{e,f,g}\models_\setseqsingleOmega D_{\{a,b,c\}}z$. The communication the agents can perform is as follows:
\begin{itemize}
    \item In the time steps before $\omega$, agent $a$ tells the other two exactly which $p_{i,j}$ and $q_{i,j}$ hold.
    \item At time $\omega$, agents $b$ and $c$ know which $p_{i,j}$ and $q_{i,j}$ hold. As a consequence, $b$ knows which $x_j$ hold while $c$ knows which $y_j$ hold.
    \item In the time steps between $\omega$ and $2\times \omega$, agents $b$ and $c$ tell the other two which $x_j$ and $y_j$ hold.
    \item At time $2\times \omega$, all agents know exactly which $x_j$ and $y_j$ hold, and therefore whether $z$ holds.
    \item We assumed we were in a \rev{world} where $z$ is true, so the agents learn that $z$ is true.
\end{itemize}
Now, to show that $\M,s_{e,f,g}\not\models_\setsimsingle D_{\{a,b,c\}}z$.

The key observation here is that $b$ and $c$ cannot provide any non-trivial announcements. Suppose towards a contradiction that $\M,s_{e,f,g}\models \square_b\psi_b$ and $\M,s_{e',f',g'}\not\models \psi_b$. Let $Q$ be the set of atoms in $\psi_b$. Then there is some $s_{e'',f,g''}$ that is $Q$-bisimilar to $s_{e',f',g'}$. But, by $\M,s_{e,f,g}\models \square_b\psi_b$, we have $\M,s_{e'',f,g''}\models \psi_b$, which by bisimilarity implies $\M,s_{e',f',g'}\models \psi_b$, contradicting our assumption. 

Hence the only announcements $b$ can provide hold in every \rev{world} of the model, and are therefore uninformative. That $c$ cannot make a non-trivial announcement is shown similarly.

This means that only $a$ can provide information, so after the communication all \rev{worlds} of the form $s_{e,f',g'}$ are still accessible. Each agent then still considers both $z$ \rev{worlds} and $\neg z$ \rev{worlds} to be possible, so $z$ is not distributed knowledge.
\end{proof}

\bibliographystyle{plain-alt}
\bibliography{references}

\end{document}